\documentclass{fundam}

\publyear{25}
\papernumber{5}
\volume{195}
\issue{1-4}
\theDOI{10.46298/fi.12354}

\versionForARXIV

\usepackage{amsfonts,latexsym,amssymb}
\usepackage{newtxtext}          %
\usepackage{newtxmath}
\usepackage[T1]{fontenc}
\usepackage[utf8]{inputenc}
\usepackage{enumitem}
\usepackage{bbm}
\usepackage{xspace}
\usepackage[all]{xy}
\usepackage{tikz}

\usepackage[numbers,sort&compress]{natbib}
\usepackage{geometry}

\DeclareMathOperator{\At}{At}
\DeclareMathOperator{\Eq}{\mathbf{Eq}}

\renewcommand{\phi}{\varphi}
\newcommand{\klam}[1]{\ensuremath{\langle #1 \rangle}}
\newcommand{\poss}[1]{\ensuremath{\klam{ #1 }}}
\newcommand{\set}[1]{\ensuremath{\{#1\}}}
\newcommand{\z}{\emptyset}
\newcommand{\Cm}{\ensuremath{\operatorname{{\mathsf{Cm}}}}}
\newcommand{\Cf}{\ensuremath{\operatorname{{\mathsf{Cf}}}}}
\newcommand{\Em}{\ensuremath{\operatorname{{\mathsf{Em}}}}} 

\newcommand{\A}{\ensuremath{A}\xspace}
\newcommand{\B}{\ensuremath{B}\xspace}
\newcommand{\C}{\ensuremath{\mathbb{C}}\xspace}
\newcommand{\F}{\ensuremath{\mathfrak{F}}\xspace}
\newcommand{\V}{\ensuremath{\mathfrak{V}}\xspace}
\newcommand{\W}{\ensuremath{\mathfrak{W}}\xspace}
\newcommand{\Wframe}{\ensuremath{\mathtt{W}}\xspace}
\newcommand{\BW}{\ensuremath{B_\mathtt{W}}\xspace}
\newcommand{\BF}{\ensuremath{B_\mathtt{F}}\xspace}
\newcommand{\FF}{\ensuremath{\mathtt{F}}\xspace}

\newcommand{\two}{\ensuremath{\mathord{\mathbbm{2}}}\xspace}

\newcommand{\df}{\ensuremath{\mathrel{:=}}}
\newcommand{\da}[1]{\ensuremath{\mathop{\downarrow}#1}}
\newcommand{\ua}[1]{\ensuremath{\mathop{\uparrow}#1}}
\newcommand{\uaP}[1]{\ensuremath{\mathop{\uparrow^{\neq}}#1}}
\newcommand{\restrict}{\ensuremath{\mathrel{\upharpoonright}}}
\newcommand{\conv}[1]{#1^{\smile{}~}}

\newcommand{\tand}{\text{ and }}
\newcommand{\tor}{\text{ or }}
\newcommand{\tiff}{if and only if\xspace}
\newcommand{\aright}{``$\Rightarrow$'': \ }
\newcommand{\aleft}{``$\Leftarrow$'': \ }

\newcommand\mrel{\mathrel}
\newcommand{\Implies}{\ensuremath{\mathrel{\Rightarrow}}}

\newcommand{\Va}{\ensuremath{\mathbf{V}}\xspace}
\newcommand{\VaCl}{\ensuremath{\mathbf{V_{Cl}}}\xspace}

\newcommand{\onto}{\twoheadrightarrow}
\newcommand{\into}{\hookrightarrow}
\newcommand{\ginto}{\overset{g}{\into}}
\newcommand{\bonto}{\overset{b}{\onto}}

\newcommand{\wlg}{w.l.o.g.\xspace }

 \parskip=\medskipamount
\numberwithin{equation}{section}

\title{The Fork and its Role in Unification of Closure Algebras}
 
 \date{}
 
 \author{Ivo D\"untsch\\
Dept. of Computer Science \\
Brock University\\	
St Catharines, Ontario, 
Canada \\
\href{mailto:duentsch@brocku.ca}{duentsch@brocku.ca} \and
Wojciech Dzik \\
Institute of Mathematics \\
University of Silesia \\
Katowice, Poland \\ \href{mailto:wojciech.dzik@us.edu.pl}{wojciech.dzik@us.edu.pl}
 }
 
 \runninghead{I. D\"untsch, W. Dzik}{The fork and its role in unification of closure algebras}
 
\begin{document}
\thispagestyle{empty}
 \maketitle

 \begin{abstract}
\noindent   We consider the two-pronged fork frame $\FF$ and the variety $\Eq(\BF)$ generated by its dual closure algebra $\BF$. We describe the finite projective algebras in $\Eq(\BF)$ and give a purely semantic proof that unification in $\Eq(\BF)$ is finitary.   
The splitting of the lattice of varieties of closure algebras given by the subdirectly irreducible algebra $\BF$ separates varieties with unitary from varieties with finitary unification type: All varieties with finitary type contain $\BF$, and all varieties with unitary type are contained in the splitting companion of $\BF$.
 \end{abstract}
 
 \section{Introduction}
 
Unification of (first order) terms was introduced by J.A.Robinson as the basic operation of the resolution principle used in automated theorem provers. Nowadays unification plays an essential role in many applications of logic to Computer Science.
  Unification theory for equational theories and logics is an important topic, for example, for automatic deduction and rewriting systems. It  has been a well researched concept for some time;  for details we refer the reader to the chapter by \citet{bs01} in the Handbook of Automated Reasoning, and for an early account of the role of unification in Computer Science we point the reader to the overview by \citet{bur95}. 
 Algebraic unification via projective algebras was introduced and investigated by \citet{ghi97}. He showed that for any equational theory the symbolic and the algebraic unification types coincide via a suitable translation. He  concludes that, broadly speaking, unification in varieties depends only on its finitely presented  projective algebras. The syntactic approach to unification often requires long and complicated calculations, while the semantic approach via projective algebras and their dual injective frames often offers much simpler  solutions. It is particularly useful for the determination of the unification type of a given theory or a logic.  Another important aspect is the effectiveness of this determination. 

In this paper the  algebraic approach via projective algebras is demonstrated to determine the unification type of varieties of closure algebras. 
In particular, we investigate in detail the equational theory of the complex algebra $\BF$ of the 2-pronged fork $\FF$, shown in Figures~\ref{fig:fork} and \ref{fig:forkA}. Elements closed in $\BF$ are shown as bullets $\bullet$.
  
 \begin{figure}[htb]
 \begin{minipage}[t]{0.48\textwidth}
 \vspace{0pt}
   \caption{The 2-pronged fork $\FF$}\label{fig:fork}
$$
\xymatrix{
v &  & w \\
 &u \ar@{->}[ru] \ar@{->}[lu] &
}
$$
 \end{minipage}
 \begin{minipage}[t]{0.48\textwidth}
 \caption{The complex algebra $\BF$}\label{fig:forkA}
 \vspace{-13pt}
\begin{gather*}
\xymatrix{
& \bullet \ar@{-}[rd] \ar@{-}[ld] \ar@{-}[d] & \\
\bullet \ar@{-}[d] \ar@{-}[rd] & \circ \ar@{-}[rd] \ar@{-}[ld]  & \bullet \ar@{-}[ld] \ar@{-}[d] \\
\circ \ar@{-}[rd] & \bullet \ar@{-}[d] & \circ  \ar@{-}[ld] \\
& \bullet &
}
\end{gather*}
 \end{minipage}
\end{figure}
 The 2-pronged fork, or simply, fork, and its logic have received some interest in the literature: \citet{avbb2003} have provided an axiomatization of $\Eq(\BF)$,  \citet{dzik06} proved that it does not have unitary unification, and recently \citet{dkw22} showed  that the logic has finitary unification. The unification results were proved by $n$-Kripke models -- which are variants of usual Kripke models --  and syntactic means. In this paper we provide semantic proofs of these results using the duality between finite closure algebras and finite quasiorders; along the way we provide a characterization of the finite projective algebras in $\Eq(\BF)$. We also give an example of an algebra in $\Eq(\BF)$ which has finitary unification. As the algebraic approach to unification  may not be familiar to all readers, we shall briefly outline its concepts tailored to our situation in Section \ref{sec:alguni}.
 
\section{Notation and first definitions}\label{sec:def}

A \emph{frame} is a structure $\klam{W,R}$ where $W$ is a non-empty set and $R$ a binary relation on $W$;  we denote the converse of $R$ by $\conv{R}$. If $x \in W$ we set $R(x) \df \set{y \in W: x \mathrel{R} y}$. A non-empty subset $U$ of $W$ is \emph{connected} (in the sense of graph theory) if for all $x,y \in U$ there is an $R \cup \conv{R}$-path from $x$ to $y$. A maximally connected subset of $W$ is called a \emph{connected component}, or just \emph{component}. $R$ is \emph{rooted} if there is some $x \in W$ such that there is an $R$-path from $x$ to every element of $W$. 
  
 A \emph{generated subframe} of $\klam{W,R}$ is a structure $\klam{V,S}$ such that $\klam{V,S}$ is a first order substructure of $\klam{W,R}$ and satisfies
 \begin{gather}\label{def:gensubframe}
(\forall u,v \in W)[u \in V \tand u\mrel{R} v \ \Implies v \in V].
 \end{gather}
 If $\klam{V,S}$ is (isomorphic to) a generated substructure of $\klam{W,R}$ we write $\klam{V,S} \ginto \klam{W,R}$. 
 
A \emph{bounded morphism} is a  mapping $p\colon W \to V$ which preserves $R$ and satisfies the \emph{back condition}
\begin{gather}
p(x) \mathrel{S} z \Implies  (\exists y \in W)[x \mathrel{R} y \tand p(y) = z].
\end{gather}
If $p$ is onto we write $\klam{W,R} \bonto \klam{V,S}$.

Suppose that $\F$ is a class of frames. We call $\klam{V,S} \in \F$ \emph{injective with respect to \F},\label{injective} if for every $\klam{W,R} \in \F$ and every injective bounded morphism $q\colon V \into W$ there is some surjective bounded morphism $p\colon W \onto V$ such that $p \circ q$ is the identity on $V$.

Let $\precsim $ be a quasiorder on $W$, i.e. reflexive and transitive. We set $\da{x} \df \set{y \in W: y \precsim x}$, $\ua{x} \df \set{y \in W: x \precsim y}$, and $\uaP{x} \df \set{y \in W: x \precnsim y}$; we will index $\precsim$ if necessary. The relation $\theta_\precsim  \df \set{\klam{x,y}: x\precsim y \tand y\precsim x}$ is an equivalence relation whose set of classes  can be partially ordered by $x/\theta \leq_\precsim  y/\theta$ \tiff $x \mathrel{\precsim } y$; the classes of $\theta_\precsim $ are called \emph{clusters}. A non-empty $M \subseteq W$ is called an \emph{antichain}, if all elements of $M$ are pairwise incomparable. $M$ is called \emph{dense} or  \emph{complete}, if each element of $W$ is below some element of $M$ with respect to $\precsim$.

The \emph{height} $h(W)$ of a finite $W$ is the length of a longest chain of clusters,\footnote{~Also called \emph{depth} or \emph{rank} \cite[p. 46]{bs84}.
 } 
and the \emph{width} $w(W)$ of $W$ is the cardinality of the largest antichain; observe that the elements of an antichain come from different clusters. The \emph{local width} $lw(W)$ of $W$ is $\max\set{w(\ua{x}): x \in W}$. 
 
  A \emph{$\mu$-set} is a dense antichain of $W$. It was shown by \citet{ghi97} that all $\mu$ sets of $W$ have the same cardinality. We say that \emph{$\klam{W,\precsim}$ is of type}

\begin{description}[font=\normalfont,nosep]
\item[\emph{unitary} ($1$)] if $W$ has a $\mu$ set of cardinality $1$,
\item[\emph{finitary} ($\omega$)] if $W$ has a finite $\mu$ set with more than one element, 
\item[\emph{infinitary} ($\infty$)] if $W$ has an infinite $\mu$-set, 
\item[\emph{nullary} ($0$)] if $W$ has no $\mu$ set. 
\end{description}
If no confusion can arise we will identify algebras with their universe. An algebra $A$ is called \emph{directly indecomposable} if $A$ is not isomorphic to a direct product of two non-trivial algebras. 

A variety \Va is called 
\begin{enumerate}
\item \emph{finite} or \emph{finitely generated}, if it is generated by a finite algebra, equivalently, if it is generated by a finite set of finite algebras.
\item \emph{locally finite}, if every finitely generated $B \in \Va$ is finite.  
\end{enumerate}
It is well known that a finite variety is locally finite, see e.g. \cite[Theorem 10.16]{bs_ua}.

 A \emph{closure operator} on a Boolean algebra $\klam{B, +, \cdot, -, 0, 1}$ with natural order $\leq$ is a mapping $f\colon B \to B$ which satisfies $f(0) = 0$ and for all $a,b \in B$ 
 \begin{align*}
 &f(a+b) = f(a) + f(b), \\
 & a \leq f(a), \\
 &f(f(a)) \leq f(a).
 \end{align*}
$b \in B$ is called \emph{closed}, if $b = f(b)$. In this case, the principal ideal $\da{b}$ is closed under $f$, and the mapping $B \onto \da{b}$ defined by $p(a) \df b \cdot a$ is a homomorphism; $\da{b}$ is called the \emph{relative algebra of $B$ with respect to $b$}. Its largest element is $b$, join, meet, $0$ and the modal operator are inherited from $B$, and complements are taken relative to $b$. If $b$ and $-b$ are closed, then $B \cong \da{b} \times \da{-b}$.

The \emph{dual mapping} of $f$ is denoted by $f^\partial$  and defined by $f^\partial(a) \df -f(-a)$; such operator is called \emph{interior operator}. An element $b \in B$ is called \emph{open}, if $b = f^\partial(b)$, and \emph{clopen}, if it is closed and open.

 A \emph{closure algebra} is a structure $\klam{B,f}$, where $B$ is a Boolean algebra and $f$ is a closure operator; the pair $\klam{B,f^\partial}$ is called an \emph{interior algebra}. The set of atoms of $B$ is denoted by $\At(B)$. The variety of closure algebras is denoted by $\VaCl$. In the rest of the paper we suppose that $\Va$ is a non-trivial variety of closure algebras unless otherwise indicated. The identity is the unique closure operator on the two-element Boolean algebra with universe $\set{0,1}$, and we denote this algebra by \two. Its variety $\Eq(\two)$ is the smallest non-trivial variety of closure algebras. 

$B \in \Va$ is called \emph{projective in \Va} \tiff for every $A \in \Va$ and every surjective homomorphism $p\colon A \onto B$ there is some injective homomorphism $q\colon B \into A$ such that the composition $p \circ q$ is the identity on $B$. In this situation, $B$ is called a \emph{retract of $A$}, and $p$ is called a \emph{retraction}.  $A$ is called \emph{injective in \Va} \tiff it is a retract of each of its extensions in \Va.\footnote{~Since categorical epimorphisms in \VaCl are onto by a result of \citet{nem83} the notions of weak projectivity of \cite{bd74} and projectivity as well as of weak injectivity and injectivity coincide. Thus, epimorphisms are exactly the surjective homomorphisms in \VaCl. For the situation in varieties of Heyting algebras see \cite{mw20}.}

Note that \two is projective in \Va. If \Va is locally finite, the projectivity of a finite algebra depends only on finite algebras: 

\begin{lemma}\label{lem:fin}
Suppose that $\Va$ is locally finite and that $\B \in \Va$ is finite. Then $\B$ is projective in $\Va$ \tiff\ for every finite $A \in \Va$, every surjective homomorphism $A \onto \B$ is a retraction.
\end{lemma}
\begin{proof}
\aright This is clear.

\aleft Suppose that $p\colon \A \onto \B$ is a surjective homomorphism. For every $b \in B \setminus \set{0,1}$ choose some $a_b \in p^{-1}(b)$, and let $\A'$ be the subalgebra of $\A$ generated by $\set{a_b: b \in B \setminus \set{0,1}}$. Since $\Va$ is locally finite, $\A'$ is finite. Clearly, $p' \df p\restrict \A'$ is a surjective homomorphism $\A' \onto \B$. By the hypothesis, there is some $q\colon B \into \A' \subseteq \A$ such that $p' \circ q = id(B)$, and it follows from the choice of $\A'$ that $p \circ q = id(B)$.
\end{proof}
If $A$ is a homomorphic image of $B$ we write $B \onto A$, and if $B$ is isomorphic to a subalgebra of $A$ we write $B \into A$. If $p:A \to B$ is a homomorphism, its \emph{kernel} is the set $f^{-1}(0)$, denoted by $\ker(p)$. It is well known that $\ker(p)$ is a closed ideal, i.e. $f[\ker(p)] \subseteq \ker(p)$, and that each congruence is determined by the closed ideal of elements congruent to $0$.

The \emph{canonical frame} of $\klam{B,f}$  is the structure $\Cf(\B) \df \klam{W_B, \precsim_f}$ where $W_B$ is the set of ultrafilters of $B$, and $\precsim_f$ is the binary  relation on $W_B$ defined by 
\begin{gather}\label{def:Rf}
F \mathrel{\precsim_f} G \text{ \tiff } f[G] \subseteq F.
\end{gather}
If $f$ is understood we shall usually omit the subscript. Conversely, if $\W \df \klam{W,R}$ is a frame, its \emph{complex algebra} is the structure $\Cm(\W) \df \klam{2^W, \poss{R}}$, where $2^W$ is the power set algebra of $W$ and $\poss{R}\colon 2^W \to 2^W$ is the mapping defined by 
\begin{gather*}
\poss{R}(X) \df \set{x \in W: R(x) \cap X \neq \z}.
\end{gather*}
We denote $\Cm\Cf(\B)$ by $\Em(\B)$ and call it the \emph{canonical embedding algebra} of $B$. The mapping $h\colon \B\to \Em(\B)$, defined by $h(a) \df \set{U \in W_B: a \in U}$, is an embedding, and $\B \cong \Em(\B)$ \tiff $\B$ is finite. Furthermore, $\precsim_f$ is a quasiorder, and, if $R$ is a quasiorder, then $\poss{R}$ is a closure operator \cite{jt51}. 

The following facts are decisive, see e.g. \cite[Theorem 5.47]{brv_modal}:

\begin{lemma}\label{lem:dual} Suppose that $\A, \B$ are closure algebras and $\V, \W$ are frames. Then,
\begin{enumerate}
\item If $\A \into \B$, then $\Cf(\B) \bonto \Cf(\A)$.
\item If $\A \onto \B$, then $\Cf(\B) \ginto \Cf(\A)$.
\item If $\V \ginto \W$, then $\Cm(\W) \onto \Cm(\V)$.
\item If $\V \bonto \W$, then $\Cm(\W) \into \Cm(\V)$.
\end{enumerate}
\end{lemma}
We shall use this duality throughout without further reference. If the structures considered are finite, then $\Em(\B) \cong \B$, $\Cf\Cm(\W) \cong \W$, and the implications above are, in fact, equivalences. In this situation we will work either with frames or algebras, depending on which way is more transparent. Such procedure shows the fruitful interaction of algebraic and relational semantics. Since the ultrafilters of a finite algebra $A$ correspond to its atoms we will usually identify an ultrafilter with the atom which generates it, so that $\Cf(A) = \At(A)$. In this way, $\precsim$ becomes a quasiorder on $\At(A)$ and \eqref{def:Rf} becomes
\begin{gather}\label{ordframe}
a \precsim_f b \text{ \tiff } \ a \leq f(b).
\end{gather}
Thus, for $a \in \At(A)$, $f(a) = \sum\set{b \in \At(A): b \precsim a}$. 

By the remarks on p. \pageref{injective}, the canonical frame of a finite projective algebra is injective, and the complex algebra of an injective finite frame is projective in the respective category. 

For unexplained notation and concepts the reader is invited to consult the standard textbooks by \citet{kop89} for Boolean algebras, \citet{bs_ua} for universal algebra, and \citet{brv_modal} for modal logic.

\section{Algebraic unification}\label{sec:alguni} 
We shall briefly describe the concept of algebraic unification presented by \citet{ghi97} as applicable to locally finite varieties.\footnote{~\citeauthor{ghi97} considers ``finitely presented'' instead of ``finite'' algebras, but in locally finite varieties these classes coincide.} For additional background we invite the reader to consult the notes by \mbox{\citet{bur01}} for a concise introduction to E-unification including examples.

Let $A \in \Va$ be finite. A \emph{unifier} of $A$ is a pair $\klam{u,B}$ where $B$ is finite  and projective in $\Va$, and $u\colon A \to B$ is a homomorphism.%
\footnote{~Strictly speaking we should define a unifier for $A \in \Va$ as a triple $\klam{A,u,B}$; we omit $A$ because we consider only unifiers of a fixed $A$.}
The collection of unifiers of $A$ with respect to \Va is denoted by $U_A^\Va$. $A$ is \emph{unifiable}, if $U_A^\Va \neq \z$. If \Va is understood, we omit the superscript. We remark in passing that a closure algebra $A$ is unifiable \tiff \two is a homomorphic image of $A$, see \cite{cit18}, hence, every non-trivial closure algebra is unifiable.

Given two unifiers $\klam{u,B}$ and $\klam{v,C}$ of $A$, we say that \emph{$\klam{u,B}$ is more general  than $\klam{v,C}$},%
\footnote{~The quasiorder on $U_A$ is not uniformly defined in the literature. We chose $\succcurlyeq$ to be consistent with \cite{ghi97} and the $\mu$-sets of Section \ref{sec:def}
} 
 written as $\klam{u,B}\succcurlyeq \klam{v,C}$, if there is a homomorphism $h\colon B \to C$ such that the diagram in Figure \ref{fig:orduni} commutes, i.e. that $v = h \circ u$.
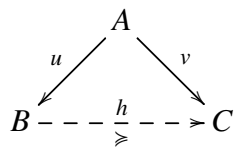
\begin{figure}[htb]
\caption{Quasiordering algebraic unifiers}\label{fig:orduni}
$$
\xymatrix{
& {A} \ar[ld]_{u} \ar[rd]^{v} \\
{B} \ar@{-->}[rr]^{h}_{\succcurlyeq} && {C} }
$$
\end{figure}
We denote the converse of $\succcurlyeq$ by $\preccurlyeq$; clearly, both relations are quasiorders. 

\noindent
If $\klam{u,B} \preccurlyeq \klam{v,C}$ and  $\klam{u,B} \succcurlyeq \klam{v,C}$ we write $\klam{u,B} \approx \klam{v,C}$, The relation $\approx$ is an equivalence relation on $U_A$, and $U_A/\approx$ can be partially ordered as described in Section \ref{sec:def}. We say that the \emph{unification type of $A$} is the unification type of the partially ordered set  $\klam{U_A/\approx}$, denoted by $t(A)$. 

By the homomorphism theorem a unifier $\klam{u,B}$ of $A$ is determined by the closed ideal $\ker(u)$ with associated congruence $\theta$, its canonical surjective homomorphism $p_\theta\colon A \onto A/\theta$,  and an embedding $e$ into $B$:
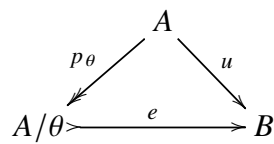
\begin{gather}\label{decomp}
\xymatrix{
& {A} \ar@{->>}[ld]_{p_\theta} \ar[rd]^{u} \\
{A/\theta}~ \ar@{>->}[rr]^{e} && {B} }
\end{gather}
In this sense, we can think of a unifier of $A$ as a triple $\klam{\theta, e, B}$ where $\theta$ is a congruence on $A$, $B$ is a finite algebra projective in $\Va$, and $e\colon A/\theta \into B$ is an embedding. Note that $A/\theta$ need not be projective, but only needs to be embeddable into a projective algebra. 
We denote by $\C(A)$ the set of  congruences of $A$ for which $A/\theta$ is isomorphic to a subalgebra of some algebra projective in \Va, also called \emph{admissible congruences}.  Using this decomposition we depict $\klam{u,B} \succcurlyeq \klam{v,C}$ in Figure \ref{fig:orduni2}. 
\begin{figure}[htb]
\caption{Quasiordering algebraic unifiers using quotients}\label{fig:orduni2}
$$
\xymatrix{
& {A} \ar[ld]_{p_\theta} \ar[rd]^{p_\psi} \\
A/\theta \ar[d]_{e_\theta} & & A/\psi \ar[d]_{e_\psi} \\
{B} \ar@{-->}[rr]^{h}_{\succcurlyeq} && {C} 
}
$$
\end{figure}
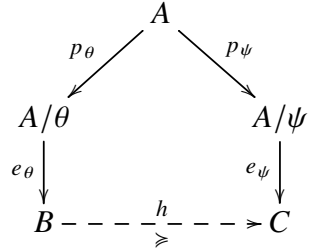

Below we collect some simple properties of unifiers which we shall use later on. 

\begin{lemma}\label{lem:retract}  
\begin{enumerate}
\item Suppose that $\klam{u,B}, \klam{v,C}$ are  unifiers of $A$. If $\klam{u,B} \succcurlyeq \klam{v,C}$, then $\ker(u) \subseteq \ker(v)$. Consequently, 
\begin{enumerate}
    \item If $\ker(u)$ and $\ker(v)$ are incomparable with respect to $\subseteq$, then $\klam{u,B}$ and $\klam{v,C}$ are incomparable with respect to $\succcurlyeq$.
    \item $\klam{u,B} \approx \klam{v,C}$ implies $\ker(u) = \ker(v)$.
\end{enumerate}
 \item If $\klam{u,B} \in U_A$, $C$ is projective in \Va,  and $B \leq C$ is a retract of $C$ with $i\colon B \into C$ the identity embedding, then, $\klam{u,B} \approx \klam{i \circ u,C}$.
\end{enumerate}
\end{lemma}
\begin{proof}
1. Suppose that $h\colon B \to C$ with $h \circ u = v$. If $u(x) = 0$, then $h(u(x)) = v(x) = 0$. 

2. Let  $p\colon C \onto B$ be a retraction:
$$
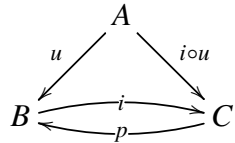
\xymatrix{
& A \ar[ld]_{u} \ar[rd]^{i \circ u} & \\
B  \ar@/^/[rr]|i
&& C \ar@/^/[ll]|{p} }
$$
Since $i(u(x)) = u(x)$ we have $\klam{u,B} \succcurlyeq \klam{i \circ u,C}$. For the converse, let $x \in A$; then, $u(x) \in B \leq C$, and $p(i(u(x))) = u(x)$, since $p \restrict B$ is the identity.
\end{proof}

Even $u[A] \cong v[A]$ does not imply that $\klam{u,B} \succcurlyeq \klam{v,B}$:

\begin{example}\label{ex:2}
Let $A \in \Va$ and $F,G$ be different closed prime ideals of $A$, and $p_F\colon A \onto A/F$ and $p_G\colon A \onto A/G$ be the canonical surjective homomorphisms; then, $A/F =  A/G = \two$. If $a \in F \setminus G$, then $0 = p_F(a) \neq p_G(a) = 1$ which shows that there is no homomorphism $h\colon \two \to \two$ such that $h \circ p_F = p_G$.
\QED
\end{example}
Generalizing unifiers to varieties, we say that \Va has unification type
\begin{description}[font=\normalfont,nosep]
\item[\emph{unitary}] if every unifiable $A \in \Va$ has type $1$.
\item[\emph{finitary}] if every unifiable $A \in \Va$ has type $1$ or type $\omega$, and there is some unifiable $A \in \Va$ with type $\omega$,
\item[\emph{infinitary}] if every unifiable $A \in \Va$ has type $1, \omega$ or $\infty$, and there is some unifiable $A \in \Va$ with type $\infty$,
\item[\emph{nullary}] if there is some $A \in \Va$ with type $0$.
\end{description}
The unification type of $\Va$ is denoted by $t(\Va)$. If we order unification types by $1 \lneq \omega \lneq \infty \lneq 0$, we see that $t(\Va) = \max\set{t(A): A \in \Va}$, see e.g. \cite[Definition 3.4]{bs01}.

There is another kind of unification which we will use in Section \ref{sec:W}: Unification of $A$ is called \emph{filtering} \cite{gs04}, if $U_A$ is directed with respect to $\preccurlyeq$, that is, for every two unifiers $\klam{u_1,B_1}$ and $\klam{u_2,B_2}$ of $A$ there is a unifier of $A$ more general than both of them. Note that if the unification of $A$ is filtering, then the unification type of $A$ is unitary or nullary. We say that unification in a variety \Va is filtering, if unification is filtering for each unifiable $A \in \Va$. \citet{gs04} have provided an algebraic characterization for unification  to be filtering:
\begin{theorem}\label{thm:filter}\cite[Theorem 3.2]{gs04} 
Unification in $\Va$ is filtering \tiff the product of any two finite  projective algebras in \Va is projective in \Va.
\end{theorem}

\section{The fork}

The \emph{$2$-pronged fork}, or simply \emph{fork}, is the frame $\FF$ shown in Figure \ref{fig:fork}. Its complex algebra is denoted by $\BF$, and the variety it generates by $\Eq(\BF)$; since $\Eq(\BF)$ is a finite variety, it is locally finite.  A \emph{fork algebra} is a nontrivial finite algebra in $\Eq(\BF)$.\footnote{~These should not be confused with the fork algebras of \citet{Frias2002} which are a definitional extension of relation algebras.} A \emph{fork frame} has the form $\Cf(B)$ for a fork algebra $B$. \citet[Theorem 5.7]{avbb2003} have shown that the variety $\Eq(\BF)$ is the variety of closure algebras which is characterized by the axioms
\begin{xalignat}{2}
& f^\partial(f(x \cdot f(-x))+x) \leq x, \tag{\textbf{Grz}} \label{grz} \\
&-x \cdot f(x) \leq f(f^\partial(x)), \tag{$\mathbf{BD_2}$} \label{bd2} \\
& -(x \cdot y \cdot f(x \cdot -y) \cdot f(-x \cdot y) \cdot f(- x \cdot -y))= 1. \tag{$\mathbf{BW_2}$} \label{bw2}
\end{xalignat}
The axiom \textbf{Grz} implies that a fork frame is a partial order \cite{bs84}, which we will denote by $\precsim$. The first order frame conditions corresponding to \eqref{bd2}, respectively, to \eqref{bw2} are
\footnote{~Computed by SQEMA \cite{Georgiev2006}.}
\begin{gather}\label{bd2R}
\forall y(x \precsim y \Implies (x = y \tor  \exists z_1(x \precsim z_1 \tand \forall z_2(z_1 \precsim z_2 \Implies y = z_2))))
\end{gather}
and
\begin{multline}\label{bw2R}
  \forall y_1 \forall y_2((x \precsim y_1 \tand  x \precsim y_2) \Implies (x = y_1 \tor  x = y_2 \tor  y_1 = y_2 \tor \\
  \forall z_1(x \precsim z_1 \Implies (x = z_1 \tor  y_1 = z_1 \tor  y_2 = z_1)))).
\end{multline}
Together, \eqref{bd2R} says that the height of a fork frame $W$ is at most two, and \eqref{bw2R}  says that the local width of $W$ is also at most two, that is,   every  $x \in W$ is related to at most two other elements. Together they imply that a rooted fork frame has one of the following forms:

$$
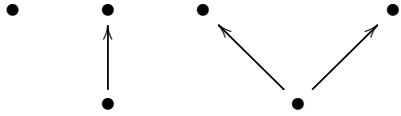
\xymatrix{
 \bullet & \bullet  & \bullet  & & \bullet \\
& \bullet \ar@{->}[u] & & \bullet\ar@{->}[lu]  \ar@{->}[ru]
}
$$

We denote by $L_W^1$ the points on the lower level and by $L_W^2$ the points on the upper level of a fork frame $W$. The points in $L_W^1$ correspond to the closed atoms of $\Cm(W)$, whereas $L_W^2$ corresponds to the set of its non-closed atoms. Algebraically, 
\begin{gather}
L_{\At(A)}^1 = \set{a \in \At(A): a = f(a)}, \quad L_{\At(A)}^2 = \set{a \in \At(A): a \lneq f(a)}. 
\end{gather}

For later use we mention the following observations:
 
 \begin{lemma}\label{lem:atom}
Suppose that $A$ is a fork algebra, 
\begin{enumerate}
\item If $a \in \At(A)$, then every element below $f(a)$ different from $a$ is closed.
\item Every atom of $A$ is open or closed.
\end{enumerate}
\end{lemma}
\begin{proof}
1. It is sufficient to show that every atom below $f(a)$ and different from $a$ is closed. Suppose that $b \in \At(A), a \neq b$ and $b \leq f(a)$. Then, $b \precsim a$ by \eqref{ordframe}, and $a \neq b$ implies that $b \in L_{\Cf(A)}^1$. It follows that $b = f(b)$.

2. If $a$ is not closed and not open, then $f^\partial(a) = 0$. Since $a$ is not closed, $-a \cdot f(a) \neq 0$. On the other hand, $-a \cdot f(a) \leq f(f^\partial(a)) = 0$ by \eqref{bd2}, a contradiction.
\end{proof}

\section{Projective fork algebras}\label{sec:forkPROJ}

Our first result gives a necessary condition for a fork algebra to be projective. We prove a slightly more general result, extending the $2$-fork to an $m$-pronged fork. Let $F_{2,m}$ be the class of finite  partial orders $\klam{P,\precsim}$ of height $2$ and local width $m \geq 2$, and $\Va_{2,m}$ be its associated variety. The levels of $\klam{P,\precsim}$ are defined as for fork frames; this makes sense, since both kinds have height two. 

\begin{theorem}
Suppose that $B \in \Va_{2,m}$ is finite and projective. Then, $B$ is directly indecomposable, and  $f(a_1) \cdot \ldots \cdot f(a_k) \neq 0$ for all non-closed atoms $a_1, \ldots a_k$, when $k \leq m$.
\end{theorem}
\begin{proof}
We shall prove the dual statement. Suppose that $\klam{V, \precsim_V}$ is the canonical frame of $B$, and assume that there are $v_1, \ldots, v_k \in L_V^2$ such that $\da{v_1} \cap \ldots \cap \da{v_k} = \z$. Choose some $x \not\in V$, and set $W \df V \cup \set{x}$, $\precsim_W \df \precsim_V \cup \set{\klam{x,v_i}: 1 \leq i \leq k} \cup \set{\klam{x,x}}$; then, $\ua_W{V} \subseteq V$, and therefore $V$ is a generated substructure of $W$. Since the height of $V$ is two, adding $x$ as above does not increase the height, hence, $\Cm(W) \in \Va_{2,m}$. Suppose that $p\colon W \bonto V$ is a bounded retraction, and $p(x) = y$. Since $p$ preserves the order, $y \precsim_V v_i = p(v_i)$ for $1 \leq i \leq k$ which contradicts the hypothesis.

Assume that $V_1, V_2$ are different connected components of $V$, and let $y_i \in V_i$ be maximal. Choose some $x \not\in V$, and set $W \df V \cup \set{x}$, $\precsim_W\df \precsim_V \cup\set{\klam{x,x}, \klam{x,y_1}, \klam{x,y_2}}$. Since $m \geq 2$, height and width of $V$ are not increased, and we can proceed as above to arrive at a contradiction.
 \end{proof}

Since $\Eq(\BF) = \Va_{2,2}$, we obtain
\begin{theorem}\label{thm:projnec}
If $\klam{B,f} \in \Eq(\BF)$ is finite and projective, then $B$ is directly indecomposable and   $f(a) \cdot f(b) \neq 0$ for all non-closed atoms $a,b \in B$. 
\end{theorem}
We now show that the conditions of Theorem \ref{thm:projnec} are sufficient for projectivity.

\begin{theorem}\label{thm:proj}
Suppose that $\klam{B,f} \in \Eq(\BF)$ is finite and directly indecomposable, and that $f(a) \cdot f(b) \neq 0$ for all non-closed atoms $a,b \in B$. Then, $B$ is projective in $\Eq(\BF)$.
\end{theorem}
\begin{proof}
Suppose \wlg that $B \neq \two$. We will use duality, and set $\klam{V, \precsim} \df \Cf(B)$; furthermore, we suppose that $V$ is a generated substructure of a fork frame $W$. Since $V$ is connected, it is contained in a component of $W$, and by mapping all points of $W$ outside this component to a maximal point of $V$, we may suppose \wlg that $W$ itself is connected; in particular, for all $x \in L_W^1$ there is some $y \in L_W^2$ such that $x \precsim y$. 

We will construct a bounded epimorphism $p\colon W \onto V$ by cases. It suffices to show that $p$ preserves $\precsim$ and satisfies the back condition on $\ua{x}$ for $x \in L_W^1$. Let $p$ be the identity on $V$. We divide  $L_W^1 \setminus V$ into three disjoint (possibly empty) sets:
\begin{align}
W_1 &\df \set{x \in L_W^1 \setminus V: \uaP{x} \subseteq V}, \\
W_2 &\df \set{x \in L_W^1 \setminus V: \uaP{x} \cap V \neq \z, \uaP{x} \cap W \setminus V \neq \z}, \\
W_3 &\df \set{x \in L_W^1 \setminus V: \uaP{x} \cap V = \z}.
\end{align}

\begin{enumerate}
\item $x \in W_1$: Here, we consider two cases:
\begin{enumerate}
\item $\uaP{x} = \set{v}$: Set $p(x) \df v$. Then, $p[\ua{x}] = \set{v}$ and clearly, $p \restrict \ua{x}$ is a bounded morphism.
\item $\uaP{x} = \set{v,w}$, $v \neq w$: Choose some $u \in \da{v} \cap \da{w} \cap V$, and set $p(x) \df u$; such $u$ exists by the hypothesis. Then, the diagram
$$
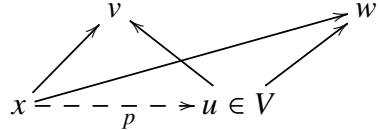
\xymatrix{
& v  &  & w \\
x \ar@{->}[ru] \ar@{->}[rrru] \ar@{-->}[rr]_p& & u \in V \ar@{->}[lu] \ar@{->}[ru]
}
$$
shows that $p\restrict[\ua{x}]$ is a bounded morphism.
\end{enumerate}
\item $x \in W_2$: Define
\begin{gather*}
Y \df \bigcup \set{\ua{x}: x \in W_2} \cap (L_W^2 \setminus V);
\end{gather*}
then,
\begin{gather*}
Y = \set{y \in L_W^2\setminus V: (\exists x \in L_W^1)[x \precsim y \text{ and } \uaP{x} \cap L_V^2 \neq \z]}.
\end{gather*}
For each $y \in Y$ we set
\begin{gather*}
X_y \df \set{x \in L_W^1: x \precsim y \tand \ua{x} \cap L_V^2 \neq \z}.
\end{gather*}
The situation $x \in X_y$ is depicted in the following diagram:
$$
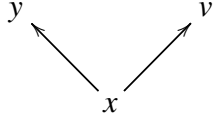
\xymatrix{
y && v  \\
& x \ar@{->}[lu] \ar@{->}[ru]
}
$$
If $x \in X_y$, then $x \not\in V$, since $x \precsim y \not\in V$ and $V$ is a generated substructure of $W$.

Our next aim is to show that $\set{X_y: y \in Y}$ is a partition of $W_2$. Assume that $y,y' \in Y, y \neq y'$, and $x \in X_y \cap X_{y'}$. Then, $x \lneq y, x \lneq y'$ and therefore, $\uaP{x} = \set{y,y'} \subseteq L_W^2 \setminus V$ by \eqref{bw2R}. This contradicts $x \in W_2$. If $x \in W_2$ there is some $y \in Y$ such that $x \precsim y$, thus, $x \in X_y$.  Hence, $p \restrict X_y$ and $p \restrict X_{y'}$  may be defined independently if $y \neq y'$.

Let $y \in Y$ and enumerate $X_y = \set{x_1, \ldots, x_k}$; then, $x_i \precsim y$ and $x_i \precsim v_i$ for exactly one $v_i \in L_V^2$; note that the $v_i$ are not necessarily different. Suppose \wlg that $v_i = v_1$ for $1 \leq i \leq m \leq k$, and set $p(y) \df v_1$ as well as $p(x_i) \df v_1$ for $1 \leq i \leq m$.

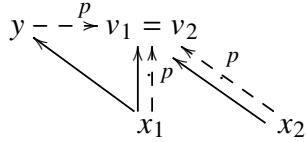
$$
\xymatrix{
y  \ar@{-->}[r]^p & v_1 = v_2 & \\
& x_1 \ar@{->}[lu] \ar@{-->}[u]_p \ar@<1ex>[u]_{.} & x_2 \ar@{-->}[lu]_p  \ar@<1ex>[lu]_{.}
}
$$

Then, $p[\bigcup \set{\ua{x_i}: 1 \leq i \leq m}] = \set{v_1}$, and clearly, $p \restrict \bigcup \set{\ua{x_i}: 1 \leq i \leq m}$ is a  bounded morphism.

For $i = m+1, \ldots, k$ choose some $u_i \in L_V^1$ such that $u_i \in \da{v_1} \cap \da{v_i}$ and define $p(x_i) \df u_i$. Now, the diagram
$$
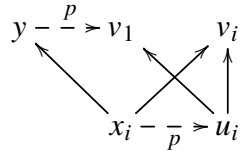
\xymatrix{
y  \ar@{-->}[r]^p & v_1 & v_i \\
& x_i \ar@{->}[lu] \ar@{->}[ru] \ar@{-->}[r]_p & u_i \ar@{->}[lu] \ar@{->}[u]
}
$$
shows that $p\restrict \ua{x_i}$ is a bounded morphism. This way we have defined $p \restrict \bigcup \set{\ua{x}: x  \in W_2}$.

\item $x \in W_3$: Thus far, we have well defined $p$ on $V$ and $\ua{(W_1 \cup W_2)}$. If $x \in \da{y} \cap \da{y'}$ for some distinct $y,y' \in L_W^2 \setminus V$, then $p$ might already been defined on $\uaP{x}$ in the previous step for example,

 $$
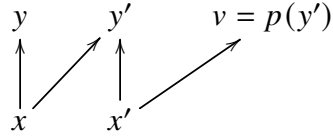
\xymatrix{
y & y' & v =  p(y') \\
x \ar@{->}[u] \ar@{->}[ru] & x' \ar@{->}[u] \ar@{->}[ru]
}
$$
First, suppose that $\uaP{x} = \set{y}$. If $p(y)$ has already been defined, set $p(x) \df p(y)$. Otherwise, choose some $v \in L_V^2$ and set $p(x) \df v, p(y) \df v$.

Finally, let $\uaP{x} = \set{y,z}$, $y \neq z$. If both $p(y)$ and $p(z)$ have been defined choose $u \in \da{p(y)} \cap \da{p(z)} \cap V$ and set $p(x) \df u$. Then,
$$
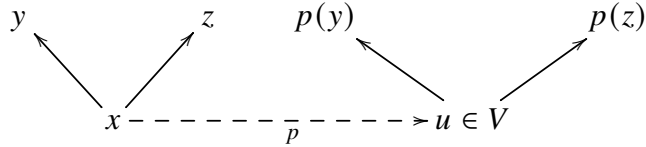
\xymatrix{
y & & z & p(y) && p(z) \\
& x \ar@{->}[ru] \ar@{->}[lu] \ar@{-->}[rrr]_p && & u \in V \ar@{->}[ru] \ar@{->}[lu]
}
$$
shows that  $p\restrict\ua{x}$ is a bounded morphism.

If only one of $p(y), p(z)$ has been defined, say, $p(y)$, choose $v \in L_V^2$, set $p(z) \df v$,  choose $u \in \da{p(y)} \cap \da{v}$ and set $p(x) \df u$:
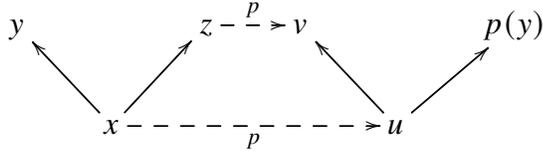
$$
\xymatrix{
y & & z \ar@{-->}[r]^p & v && p(y) \\
& x \ar@{->}[ru] \ar@{->}[lu] \ar@{-->}[rrr]_p && & u  \ar@{->}[ru] \ar@{->}[lu]
}
$$
As in the previous case, $p\restrict\ua{x}$ is a bounded morphism. If neither $p(y)$ nor $p(z)$ have been defined, choose $u \in L_V^1,  v \in L_V^2$ such that $u \precnsim v$, and set $p(y), p(z) \df v$ and $p(x) \df u$. Clearly, $p \restrict \ua{x}$ is an epimorphism.
\end{enumerate}
This completes the proof.
\end{proof}

Next, we will show that the class of projective fork algebras is closed under subalgebras. As a preparation we mention two lemmas:

\begin{lemma}\label{lem:nonclosat} Let $\klam{A,f}$ be a finite closure algebra such that $f(a) \cdot f(b) \neq 0$ for all non-closed atoms $a,b \in A$. Then, $f(a) \cdot f(b) \neq 0$ for all non-closed $a,b \in A$.
\end{lemma}
\begin{proof}
Suppose that $a,b \in A$ are non-closed, $a = a_1 + \dots + a_n$, $b = b_1 + \dots b_k$ for $a_1,  \dots, a_n, b_1, \dots, b_k \in \At(A)$. Assume that $f(a) \cdot f(b) = 0$. Since $a,b$ are not closed, there are non-closed atoms $a_i \leq a, b_j \leq b$: Otherwise, for example, $f(a) = f(a_1) + \dots + f(a_n)  = a_1 + \dots + a_n = a$ contradicting that $a$ is not closed. By the hypothesis, $f(a_i) \cdot f(b_j) \neq 0$ which implies $f(a) \cdot f(b) \neq 0$.
\end{proof}

\begin{lemma}\label{lem:di}\cite[Corollary 4.2]{nr93} 
Suppose that $\klam{A,f}$ is a non-trivial closure algebra. Then, $A$ is directly indecomposable \tiff
\begin{gather*}
(\forall x)[f(x) = x \tand f(-x) = -x \Implies x = 0 \tor x = 1].
\end{gather*}
Hence, each subalgebra of a directly indecomposable $A$ is directly indecomposable.
\end{lemma}

\begin{theorem}\label{thm:sub}
Suppose that $\klam{A,f}$ is a projective fork algebra, and $B$ is a subalgebra of $A$. Then, $B$ is projective in $\Eq(\BF)$.
\end{theorem}
\begin{proof}
Since $A$ is directly indecomposable, so is $B$ by Lemma \ref{lem:di}. Suppose that $a,b$ are non-closed atoms of $B$, and assume that $f(a) \cdot f(b) = 0$. Since $A$ is projective, the meet of any two non-closed atoms of $A$ is non-zero by \ref{thm:projnec}, and thus, the meet of any two non-closed elements of $A$ is non-zero by Lemma \ref{lem:nonclosat}. Since $B \leq A$ this also holds for $B$.
\end{proof}
This also follows from the fact that the dual conditions of Theorem \ref{thm:projnec} are preserved under surjective bounded morphisms. It is a quite strong condition, since it implies, for example, that the projective fork algebras are exactly the subalgebras of free algebras.

Next we consider a structure which will be important in our investigation of projective algebras as well as in determining the unification  type of $\Eq(\BF)$. Consider the frame \Wframe in the form of a W shown in Figure \ref{fig:W}, and its complex algebra $\klam{\BW, f_\Wframe}$. 
\begin{figure}[htb]
\caption{The frame \Wframe}\label{fig:W}
$$
\xymatrix{
t && v  && w  \\
& u  \ar@{->}[ru] \ar@{->}[lu] && u' \ar@{->}[ru] \ar@{->}[lu]
}
$$
\end{figure}
It has 5 atoms, and we identify these with the points of $\Wframe$. The action of $f_\Wframe$ on the atoms is given in Table \ref{tab:W}.
\begin{table}[ht]
\caption{The values of $f_\Wframe$ on $\At(B_\Wframe)$}\label{tab:W}
$$
\begin{array}{l|ccccc}
x & u &  u'  & t & v  & w \\
f_\Wframe(x) & u & u' & u+t & u+u'+v & u'+w
\end{array}
$$
\end{table}

It is not hard to see that $\Eq(\BW) =\Eq(\BF)$. There is an intimate connection between $B_\Wframe$ and projective fork algebras: 
\begin{theorem}\label{thm:Wproj}
Let $A \in \Eq(\BF)$ be a directly indecomposable fork algebra. Then, $A$ is projective in $\Eq(\BF)$ \tiff $B_\Wframe$ is not isomorphic to a subalgebra of $A$.%
\footnote{~We are grateful to a reviewer who spotted an error in our original proof.}
\end{theorem}
\begin{proof}
\aright Suppose that $\klam{A,f}$ is projective in $\Eq(\BF)$; then, $f(a) \cdot f(b) \neq 0$ for all non-closed $a,b \in A$ by Theorem \ref{thm:projnec} and Lemma \ref{lem:nonclosat}. If $B_\Wframe$ is a subalgebra of $A$, then Table \ref{tab:W} shows that $t$ and $w$ are non-closed and $f(t) \cdot f(w) = 0$, a contradiction.

\aleft We show the contrapositive, namely, that a directly indecomposable fork algebra $A$ which is not projective in $\Eq(\BF)$ contains an isomorphic copy of $\BW$ as a subalgebra. 

Since $A$ is not projective, there are non-closed atoms $a,b \in A$ such that $f(a) \cdot f(b) = 0$. For later use observe that this implies $a \cdot f(b) = 0$. Suppose that 
\begin{gather*}
\At(A) = \set{a,a_1, \ldots, a_n, b, b_1, \ldots, b_m, c_1, \ldots c_k}, 
\end{gather*}
and that
\begin{gather*}
f(a) = a+ a_1 + \ldots a_n, \ f(b) = b + b_1 + \ldots b_m, \ -f(a+b) = c_1 + \ldots + c_k.
\end{gather*} 
Define
\begin{multline*}
v \df -f(a+b), \ d \df -v \cdot f(v), \\  u \df d \cdot f(a), \ t \df -d \cdot f(a),  \  u' \df d \cdot f(b), \ w \df -d \cdot f(b).
\end{multline*}
Note that $f(a) = u + t, f(b) = u' + w$, and that
\begin{align*}
u &= d \cdot f(a)= -v \cdot f(v) \cdot f(a) = f(a+b) \cdot f(v) \cdot f(a) = f(v) \cdot f(a), \\
t &= -d \cdot f(a) = (v + -f(v)) \cdot f(a) = (-f(a+b) + -f(v)) \cdot f(a) = -f(v) \cdot f(a).
\end{align*}
Furthermore, $a \cdot f(v) = 0$, i.e. $a \not\leq f(v)$: If $a \leq f(v)$, then $f(v) = f(c_1) + \ldots + f(c_k)$ implies that $a \leq f(c_i)$ for some $c_i$. Since $a \neq c_i$, $a$ is closed by Lemma \ref{lem:atom}(1), a contradiction. Since $a \cdot f(b) = 0$ as well, it follows that $a \not\leq f(b) + f(v)$, in particular, $f(b) + f(v) \neq 1$. Observe that $a \cdot f(v) = 0$ implies $a \cdot d = 0$.  Similarly it can be shown that $b \cdot f(v) = 0$ and $b \cdot d = 0$. 

Set $M \df \set{u,u',t,v,w}$. Clearly, the elements of $M$ are pairwise disjoint, and $\sum M = 1$. Let $B$ be the Boolean subalgebra of $A$ generated by $M$; by the properties of $M$ this is the closure of $M$ under joins. We will show
\begin{enumerate}
\item $M \cap \set{0} = \z$.
\item $B$ is closed under $f$.
\item $B \cong \BW$.
\end{enumerate}
1. Assume $v =0$; then, $f(a) + f(b) = 1$. Together with $f(a) \cdot f(b) = 0$, this implies that $f(a)$ and $f(b)$ are clopen complementary elements not equal to $0$ or $1$, contradicting that $A$ is directly indecomposable; thus, $v \neq 0$.  

Assume $u = 0$; then, $f(a) = u + t = t = -f(v) \cdot f(a)$, which implies $f(a) \cdot f(v) = 0$. Since $f(a) \cdot f(b) = 0$ by the assumption, we have $f(a) \cdot f(v+b) = 0$. Noting that $f(a) + f(v+b) = 1$ and $\set{f(a), f(v+b)} \cap \set{0,1} = \z$, we see that  $f(a)$ and $f(v+b)$ are clopen complementary elements not equal to $0$ or $1$, contradicting that $A$ is directly indecomposable. Similarly, $u' \neq 0$.

Assume $t = 0$; then, $f(a) = u = f(v) \cdot f(a)$, hence, $f(a) \leq f(v)$. This contradicts $a \cdot f(v) = 0$. Similarly, $w \neq 0$.

2. For $f(v)$, note that $f(v) = v \cdot f(v) + -v \cdot f(v) = v + d$, and 
\begin{gather*}
d + v = d \cdot -v + v = d \cdot (f(a) + f(b)) + v = d \cdot f(a) + d \cdot f(b) + v = u + u' + v,
\end{gather*}
which shows that $f(v) \in B$.

 Next, we show that $f(t) = f(a)$: Since $a \leq -f(v)$ we obtain $a \leq -f(v) \cdot f(a) = t$, consequently, $f(a) \leq f(t)$, and together with $t \leq f(a)$ we obtain $f(t) = f(a)$; similarly, $f(b) = f(w)$. If $a \leq d$, then $a \leq -v \cdot f(v) \leq f(v)$, contradicting $a \cdot f(v) = 0$. Since $u \neq 0$, $a \cdot d = 0$ implies that $d \cdot f(a) = a_{i_1} + \ldots + a_{i_r}$ for some $a_{i_1}, \ldots a_{i_r} \in \set{a_1, \ldots, a_n}$. By Lemma \ref{lem:atom}(1), $u = d \cdot f(a)$ is closed; similarly, $u'$ is closed. 

Altogether, we have shown that $B$ is closed under $f$, hence, a (modal) subalgebra of $A$.

3. The action of $f$ on the atoms of $B$ is given by
$$
\begin{array}{l|ccccc}
x & u &  u'  & t & v  & w \\
f(x) & u & u' & f(a) & v +d & f(b).
\end{array}
$$
Noting that $f(a) = d \cdot f(a) + -d \cdot f(a) = u + t$, $f(b) = d \cdot f(b) + -d \cdot f(b) = u' + w$, and $f(v) = v+d = u + u' +v$ we see that $B \cong \BW$. 
\end{proof}
Collecting the previous results, we arrive at characterizations of projective fork algebras:

\begin{theorem}\label{thm:projchar}
Suppose that $\klam{A,f}$ is a fork algebra. Then, the following conditions  are equivalent:
\begin{enumerate}
\item $A$ is projective in $\Eq(\BF)$. 
\item $A$ is directly indecomposable and $f(a) \cdot f(b) \neq 0$ for all non-closed elements of $A$. Dually, $\Cf(A)$ is connected, and $ \da{x} \cap \da{y} \neq \z$ for all $x,y \in L_{\Cf(A)}^2$.
\item $A$ is directly indecomposable and $B_\Wframe$ is not a subalgebra of $A$.
\end{enumerate}
\end{theorem}

Finally in this section we present a general result regarding projectivity of \BF in $\VaCl$. Its proof requires some background of Heyting algebras and interior algebras which we will not go into. We shall just give the references leaving the details to the interested reader.

\begin{theorem}\label{thm:BFproj}
$\BF$ is projective in \VaCl.
\end{theorem}
\begin{proof}
It can be seen from Figure \ref{fig:forkA} that \BF is generated by its open elements, and thus, it is a $^*$-algebra in the sense of \citet[Definition I.2.14]{blok76}. Its Heyting algebra $\BF^\circ$ of open elements is a four element Boolean algebra with a new largest element added, and therefore, $\BF^\circ$ is a projective Heyting algebra by \cite[Theorem 4.10]{bh70}. It now follows from \cite[Theorem I.7.14]{blok76} that $\BF$ is projective in \VaCl.
\end{proof}

\section{The unification type of the variety generated by the fork}\label{sec:W}
It was shown by \citet[Corollary 4.8]{dkw23} by means of Kripke models that $\Eq(\BF)$ has finitary unification. In this section we shall present a much simpler proof of this result by algebraic means.

\begin{theorem}\label{thm:lfuni}
If \Va is a locally finite variety in which the class of finite projective algebras is closed under subalgebras, then unification in \Va is either unitary or finitary.\footnote{~We are grateful to a reviewer for suggesting this generalization of our original result.
}
\end{theorem}
\begin{proof}
Suppose that $A \in \Va$ is finite, and  that $\klam{u,B}$ is a unifier of $A$. Then, there are some admissible $\theta \in \C(A)$ such that $p_\theta[A] = u[A]$, and some embedding $e\colon p_\theta[A] \into B$ such that $u = e \circ p_\theta$. Since $A/\theta$ is isomorphic to a subalgebra of the projective algebra $B$, it is projective itself by the hypothesis; therefore, $\klam{p_\theta, A/\theta}$ is also a unifier of $A$ and $\klam{p_\theta, p_\theta[A]} \succcurlyeq \klam{u,B}$:
    
$$
\xymatrix{
& {A} \ar@{->>}[ld]_{p_\theta} \ar[rd]^{u = e \circ p_\theta} \\
{A/\theta} \ar@{->}[rr]^{e}_{\succcurlyeq} && {B} }
$$
Thus, with respect to $\preccurlyeq$, each unifier of $A$ is below a unifier of the form $\klam{p_\theta,p_\theta[A]}$ for some $\theta \in \C(A)$. As $\C(A)$ is finite, there can be no infinite $\mu$ set. 

Let $M$ be a maximal antichain in $\C(A)$. Then, $\set{\klam{p_\theta,p_\theta[A]}: \theta \in M}$ is a $\mu$ set:  Let $\klam{u,B} \in U_A$ and $u = e \circ p_\theta$; then, $\theta \in \C(A)$ and $\klam{p_\theta, p_\theta[A]} \succcurlyeq \klam{u,B}$. The maximality of $M$ now implies that $\klam{p_\psi, p_\psi[A]} \succcurlyeq \klam{p_\theta, p_\theta[A]}$ for some $\psi \in M$.
\end{proof}

\begin{theorem}\label{thm:forkuni}
The unification type of $\Eq(\BF)$ is finitary.
\end{theorem}
\begin{proof}
By Theorem \ref{thm:lfuni} it is enough to show that $\Eq(\BF)$ does not have unitary unification.
Consider the frame \Wframe of Figure \ref{fig:W} and its complex algebra \BW indicated in Table \ref{tab:W}. Examination of the frame \Wframe shows that there are exactly two connected generated subframes with more than one element, namely, $\ua{u'}$ and $\ua{u}$, corresponding to the quotients $\BW^1 \df \BW/\da{f(t)}$ and $\BW^2 \df \BW/\da{f(w)}$. Let $p_i\colon \BW \onto \BW^i$ be the respective quotient mappings; then, $\klam{p_1, \BW^1}$ and $\klam{p_2,\BW^2}$ are unifiers of $\BW$. Since $\da{f(t)} \cap \da{f(w)} = \set{0}$, it follows immediately from Lemma \ref{lem:retract} that they are incomparable with respect to $\succcurlyeq$.

 It remains to show that $\klam{p_1, \BW^1}$ and $\klam{p_2,\BW^2}$ are maximal with respect to $\succcurlyeq$. Assume that $\klam{u,A} \in U_{\BW}^\Va$ such that $\klam{u,A} \succcurlyeq \klam{p_1, \BW^1}, \klam{p_2,\BW^2}$. By Lemma \ref{lem:retract} we have $\ker(u) \subseteq \ker(p_{f(t)}) \cap \ker(p_{f(w)}) = \set{0}$. It follows that $u$ is injective, and therefore, $\BW \leq A$. Since  $A$ is projective in $\Eq(\BF)$, this contradicts Theorem \ref{thm:projchar}.
\end{proof}

Finally in this section we show that every variety of closure algebras with finitary unification contains \BF. The proof uses filtering unification defined in Section \ref{sec:alguni} and the notion of splitting pairs. The concept of a splitting pair of a lattice was introduced by \citet{mck72} and applied to the lattice of subvarieties of \VaCl by \citet{blok76}.  A pair $\klam{\Va_1, \Va_2}$ of subvarieties of $\VaCl$ is called \emph{splitting}, if 
\begin{enumerate}
\item $\Va_1 \not\leq \Va_2$,
\item If $\Va' \leq \VaCl$, then $\Va_1 \leq \Va'$ or $\Va' \leq \Va_2$.
\end{enumerate}
It is well known that $\Va_1$ is generated by a finite subdirectly irreducible algebra, say, $B$, and that $\Va_2$ is the largest subvariety of $\VaCl$ not containing $B$, called the \emph{splitting companion of $B$}; for details see e.g. \cite{blok76}. It has been known for some time \cite[Example, p. 158]{rau80}  that the splitting companion of $\BF$ is the variety $\Va_G$ of closure algebras $\klam{B,f}$ that satisfy the Geach identity
\begin{gather}\label{G}\tag{\textbf{G}}
f(f^\partial(x)) \leq  f^\partial(f(x)).
\end{gather}
It was shown by \citet{dl59} that these algebras are the algebraic models of the logic \textbf{S4.2} which is also known as \textbf{S4G}. 

To prove our theorem we need one more result concerning varieties whose unification is filtering:

\begin{lemma}\label{lem:filt} \cite[Theorem 8.4]{gs04} 
Unification for a variety \Va of closure algebras is filtering \tiff $\Va \leq \Va_G$. 
\end{lemma}

\begin{theorem}\label{thm:BFnotUni}
If $\Va$ has finitary unification, then $\BF \in \Va$.
\end{theorem}
\begin{proof}
Suppose that $\Va$ has finitary unification. Then, there is some $A \in \Va$ whose set of unifiers is not directed, that is, unification for $A$ is not filtering. From Lemma \ref{lem:filt}  we obtain $\Va\not\leq \Va_G$, and thus, $\Eq(\BF) \leq \Va$ by the splitting result.
\end{proof}

\section{The weak disjunction property}

In this section we shall give another characterization of varieties containing \BF, and relate it to the unification type of a variety of closure algebras. 

Recall, e.g. from \cite[\S 11]{bs_ua}, that for a term $\tau(x_1, \ldots, x_n)$, $\Va \models \tau(x_1, \ldots, x_n) \approx 1$ \tiff for all $A \in \Va$ and all $a_1, \ldots, a_n \in A$ we have $\tau(a_1, \ldots, a_n) = 1$. If $\klam{A,f} \in \Va$ we let $A_{\two}$ be the Boolean reduct of $\klam{A,f}$ augmented by the identity operator; then, $A_{\two} \in \Eq(\two)$, and each $A \in \Eq(\two)$ has this form. We say that \Va satisfies the  \emph{weak disjunction property} (WDP) \cite{dzik06} if for all terms $\tau_1(x_1, \ldots, x_n)$, $\tau_2(y_1, \ldots, y_k)$ in the language of \Va,
\begin{multline}\label{WDP}\tag{\textbf{WDP}}
\Va \models f^\partial(\tau_1(x_1, \ldots, x_n)) + f^\partial(\tau_2(y_1, \ldots, y_k)) \approx 1 \Implies \\ \Eq(\two) \models \tau_1(x_1, \ldots, x_n) \approx 1 \tor \Eq(\two) \models \tau_2(y_1, \ldots, y_k) \approx 1.
\end{multline}
The WDP is a weakening of the Disjunction Property which on the right side of the implication has $\Va \models \tau_1(x_1, \ldots, x_n) \approx 1 $ or $\Va \models \tau_2(y_1, \ldots, y_k) \approx 1$.%
\footnote{~Both in (DP) and (WDP) strings $x_1, \ldots, x_n$ and $y_1, \ldots, y_k$ may contain variables in common. (DP) with disjoint variables $x_i$ and $y_j$ is related (in \VaCl) to Halld{\'e}n completeness, see \cite[Section 15]{cz97}. To the best of our knowledge the Disjunction Property for the modal logic \textbf{S4}  was first proved by \citet[Theorem $\lambda$]{rs54}. For an overview of the Disjunction Property in modal logics see \cite[Chapter 15]{cz97}.} 

The next result shows that the WDP is equivalent relative to \VaCl to the axioms characterizing  $\Eq(\BF)$, namely, \eqref{grz}, \eqref{bd2}, \eqref{bw2}.  

 Since an equation holds in  $\Eq(\two)$ \tiff it holds in \two (as in the equivalent variety of Boolean algebras \cite[Proposition 2.19]{kop89}), we may write the right hand side of \eqref{WDP} as 
 \begin{gather*}
 \two \models \tau_1(x_1, \ldots, x_n) \approx 1 \tor \two \models \tau_2(y_1, \ldots, y_k) \approx 1
\end{gather*} 
\begin{theorem}\label{thm:WDPFork}
\Va has the WDP \tiff $\Eq(\BF) \leq \Va$.
\end{theorem}
\begin{proof}
\aright Suppose that $\Va \le \VaCl$ has the \emph{WDP}, and assume that $\Eq(\BF) \not\leq \Va$. Then, $\Va \leq \Va_G$, and thus, $f^\partial(f(-a))) + f^\partial(f(a)) = 1$ for all $A \in \Va$ and all $a \in A$. Let $A \in \Va, A \neq \two$. Then, the WDP implies that $A_{\two} \models -x \approx 1$ or $A_{\two} \models x \approx 1$, and thus, $-a = 1$ or $a = 1$ for all $a \in A$,  contradicting that $A \neq \two$. 

\aleft  Suppose that $\Eq(\BF) \leq \Va$. We will show the contrapositive of \eqref{WDP}. Suppose that there are terms $\tau_1(x_1, \ldots, x_n), \tau_2(y_1, \ldots, y_k)$ such that neither $\two \models \tau_1(x_1, \ldots, x_n)$ nor  $\two \models \tau_2(y_1, \ldots, y_k)$. Then, there are $0,1$ tuples $\vec{p}$ and $\vec{q}$ of length $n$, respectively, $k$ such that $\tau_1(\vec{p}) = \tau_2(\vec{q}) = 0$. We need to show that $\Va \not\models f^\partial(\tau_1(x_1, \ldots, x_n)) + f^\partial(\tau_2(y_1, \ldots, y_k)) \approx 1$, that is, we need to find some $A \in \Va, a_1, \ldots, a_n, b_1, \ldots b_k$ such that $f^\partial(\tau_1(a_1, \ldots, a_n)) + f^\partial(\tau_2(b_1, \ldots, b_k)) \neq 1$.  Consider $A \df \BF$, $\klam{a_1, \ldots, a_n} \df \vec{p}$ and $\klam{b_1, \ldots, b_k} \df \vec{q}$, and assume that $f^\partial(\tau_1(\vec{p})) + f^\partial(\tau_2(\vec{q})) = 1$. Inspecting the open elements of $\BF$ in Figure \ref{fig:forkA}, we see that $f^\partial(\tau_1(\vec{p})) = 1$ or $f^\partial(\tau_2(\vec{q})) = 1$. This implies that $\tau_1(\vec{p}) = 1$ or $\tau_2(\vec{q}) = 1$, contradicting the hypothesis.
\end{proof}

Unification is related to the WDP by the following result:
\begin{lemma}\label{lem:wdpuni}
\cite[Lemma 9]{dzik06} If $\Va$ satisfies the WDP, it does not have unitary unification.
\end{lemma}
The proof in \cite{dzik06} is a simple syntactic one, and unfortunately we have not found an algebraic proof. We will use this result to obtain a converse to Theorem \ref{thm:BFnotUni}, relative to varieties of finitary type:

\begin{theorem}
If  $\BF \in \Va$, then \Va does not have unitary unification.
\end{theorem}
\begin{shortproof}
If $\BF \in \Va$, then \Va satisfies the WDP by Theorem \ref{thm:WDPFork}, hence, it does not have unary type by Lemma \ref{lem:wdpuni}. 
\end{shortproof}
The location of the unification types relative to the splitting $\klam{\Eq(\BF), \Va_{G}}$ is shown in Figure \ref{fig:split}. It is not known whether there is a variety of closure algebras of type $\infty$; we indicate this by adding a ``?''.

\begin{figure}[htb]
\caption{Splitting and unification types}\label{fig:split}
 \begin{center}
  \begin{tikzpicture}
    \draw [very thick](0,-2.5) -- (3,0) -- (1.5,1.5) --(0,0)-- cycle;
\node[right] at (0.3,0) {$\omega,\; 0, \; \infty (?)$};
    \draw [fill=black, very thick](0,-2.5) circle(0.1);
    \node[below left] at (0,-2.5) {$\Eq(\BF)$};
    \draw [fill=white, very thick](1.5, 1.5) circle(0.1);
\node[above right] at (1.5,1.5) {$\mathrm{\Va_{CL}}$};
    \draw [very thick](0,-2.75) -- (3,-0.25) -- (3,-1.75) --(1.5,-3.25) -- cycle;
\node[left] at (2.70, -1.75) {$1,\; 0$};    
\draw [fill=black, very thick](3,-0.25) circle(0.1);
       \node[above right] at (3,-0.25) {$\mathrm{\Va_G}$};    
       \draw [fill=white, very thick](1.5,-3.25) circle(0.1);
       \node[below left] at (1.5,-3.25) {$\Eq(\two)$};
  \end{tikzpicture}
\end{center}  
\end{figure}
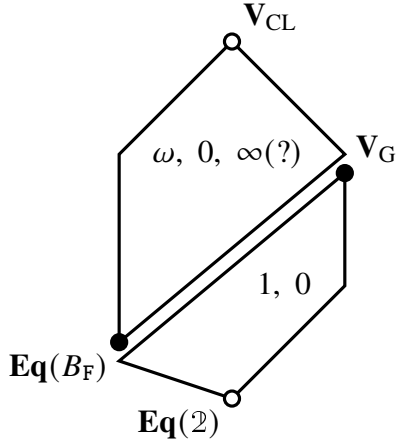 

We see that $\Eq(\BF)$ is the smallest variety of closure algebras with finitary unification. In this sense, the fork algebra $\BF$ plays the role of  a test algebra for unitary (finitary) unification in varieties of closure algebras: If \Va has unitary (finitary) unification, then it does not (does) contain $\BF$. The restriction to varieties with finitary unification is essential: \citet{dkw22,dkw23} presented infinitely many varieties of locally finite Heyting algebras with unification zero. This implies that there are infinitely many locally finite varieties of Grzegorczyk algebras which have unification type $0$ and do not contain $\BF$, and infinitely many that contain $\BF$. 

\section*{Acknowledgment}

We are very grateful to the reviewers for repeated careful reading and many suggestions which improved the manuscript considerably.

 \section*{References}
\renewcommand*{\refname}{}


\vspace{-10mm}

\begin{thebibliography}{}

\bibitem[Aiello et~al., 2003]{avbb2003}
Aiello, M., van Benthem, J., and Bezhanishvili, G. (2003).
\newblock Reasoning about space: {T}he modal way.
\newblock {\em J. Logic Comput.}, 13(6):889--920.
\newblock MR2030204.

\bibitem[Baader and Snyder, 2001]{bs01}
Baader, F. and Snyder, W. (2001).
\newblock Unification theory.
\newblock In Robinson, A. and Voronkov, A., editors, {\em Handbook of Automated
  Reasoning}, volume~1, chapter~8, pages 445--532. Elsevier.

\bibitem[Balbes and Dwinger, 1974]{bd74}
Balbes, R. and Dwinger, P. (1974).
\newblock {\em Distributive lattices}.
\newblock University of Missouri Press, Columbia, Mo.
\newblock MR0373985.

\bibitem[Balbes and Horn, 1970]{bh70}
Balbes, R. and Horn, A. (1970).
\newblock Injective and projective {H}eyting algebras.
\newblock {\em Trans. Amer. Math. Soc.}, 148:549--559.
\newblock MR256952.

\bibitem[Blackburn et~al., 2001]{brv_modal}
Blackburn, P., de~Rijke, M., and Venema, Y. (2001).
\newblock {\em Modal logic}, volume~53 of {\em Cambridge Tracts in Theoretical
  Computer Science}.
\newblock Cambridge University Press.
\newblock MR1837791.

\bibitem[Blok, 1976]{blok76}
Blok, W. (1976).
\newblock {\em Varieties of Interior Algebras}.
\newblock PhD thesis, University of Amsterdam.

\bibitem[Bull and Segerberg, 1984]{bs84}
Bull, R. and Segerberg, K. (1984).
\newblock Basic modal logic.
\newblock In {\em Handbook of philosophical logic, {V}ol. {II}}, volume 165 of
  {\em Synthese Lib.}, pages 1--88. Reidel, Dordrecht.
\newblock MR0844596.

\bibitem[Burris, 1995]{bur95}
Burris, S. (1995).
\newblock Computers and universal algebra: {S}ome directions.
\newblock {\em Algebra Universalis}, 34(1):61--71.
\newblock MR1344954.

\bibitem[Burris, 2001]{bur01}
Burris, S. (2001).
\newblock E-unification.
\newblock Available from
  \url{https://www.math.uwaterloo.ca/~snburris/htdocs/scav/e_unif/e_unif.html},
  retrieved February 17, 2024.

\bibitem[Burris and Sankappanavar, 2012]{bs_ua}
Burris, S. and Sankappanavar, H.~P. (2012).
\newblock {\em A Course in Universal Algebra - The Millenium Edition, 2012
  Update}.
\newblock Springer Ver\-lag, New York.
\newblock MR0648287.

\bibitem[Chagrov and Zakharyaschev, 1997]{cz97}
Chagrov, A. and Zakharyaschev, M. (1997).
\newblock {\em Modal logic}, volume~35 of {\em Oxford Logic Guides}.
\newblock The Clarendon Press, New York.
\newblock MR1464942.

\bibitem[Citkin, 2018]{cit18}
Citkin, A. (2018).
\newblock Projective algebras and primitive subquasivarieties in varieties with
  factor congruences.
\newblock {\em Algebra Universalis}, 79(3):Paper No. 66, 21.
\newblock MR3845737.

\bibitem[Dummett and Lemmon, 1959]{dl59}
Dummett, M. A.~E. and Lemmon, E.~J. (1959).
\newblock Modal logics between {S}4 and {S}5.
\newblock {\em Z. Math. Logik Grundlagen Math.}, 5:250--264.
\newblock MR156783.

\bibitem[Dzik, 2006]{dzik06}
Dzik, W. (2006).
\newblock Splittings of lattices of theories and unification types.
\newblock In {\em Contributions to general algebra}, volume~17, pages 71--81.
  Heyn, Klagenfurt.
\newblock MR2237807.

\bibitem[Dzik et~al., 2022]{dkw22}
Dzik, W., Kost, S., and Wojtylak, P. (2022).
\newblock Finitary unification in locally tabular modal logics characterized.
\newblock {\em Ann. Pure Appl. Logic}, 173(4):Paper No. 103072.
\newblock MR4354825.

\bibitem[Dzik et~al., 2023]{dkw23}
Dzik, W., Kost, S., and Wojtylak, P. (2023).
\newblock Unification types and union splittings in intermediate logics.
\newblock Ann. Pure Appl. Logic, to appear.

\bibitem[Frias, 2002]{Frias2002}
Frias, M.~F. (2002).
\newblock {\em Fork algebras in algebra, logic and computer science}, volume~2
  of {\em Advances in Logic}.
\newblock World Scientific Publishing Co., Inc., River Edge, NJ.
\newblock MR1922164.

\bibitem[Georgiev, 2006]{Georgiev2006}
Georgiev, D. (2006).
\newblock An implementation of the algorithm {SQEMA} for computing first-order
  correspondences of modal formulas.
\newblock Master's thesis, Sofia University, Faculty of Mathematics and
  ComputerScience.
\newblock Software homepage at
  \url{https://store.fmi.uni-sofia.bg/fmi/logic/sqema/sqema_gwt_20180317_2/index.html}.

\bibitem[Ghilardi, 1997]{ghi97}
Ghilardi, S. (1997).
\newblock Unification through projectivity.
\newblock {\em J. Logic Comput.}, 7(6):733--752.
\newblock MR1489936.

\bibitem[Ghilardi and Sacchetti, 2004]{gs04}
Ghilardi, S. and Sacchetti, L. (2004).
\newblock Filtering unification and most general unifiers in modal logic.
\newblock {\em J. Symbolic Logic}, 69(3):879--906.
\newblock MR2078928.

\bibitem[J\'{o}nsson and Tarski, 1951]{jt51}
J\'{o}nsson, B. and Tarski, A. (1951).
\newblock Boolean algebras with operators. {I}.
\newblock {\em Amer. J. Math.}, 73:891--939.
\newblock MR44502.

\bibitem[Koppelberg, 1989]{kop89}
Koppelberg, S. (1989).
\newblock {\em {G}eneral {T}heory of {B}oolean {A}lgebras}, volume~1 of {\em
  Handbook of {B}oolean algebras, edited by J.D. Monk and R. Bonnet}.
\newblock North-Holland, Amsterdam.
\newblock MR0991565.

\bibitem[McKenzie, 1972]{mck72}
McKenzie, R. (1972).
\newblock Equational bases and nonmodular lattice varieties.
\newblock {\em Trans. Amer. Math. Soc.}, 174:1--43.
\newblock MR313141.

\bibitem[Moraschini and Wannenburg, 2020]{mw20}
Moraschini, T. and Wannenburg, J.~J. (2020).
\newblock Epimorphism surjectivity in varieties of {H}eyting algebras.
\newblock {\em Ann. Pure Appl. Logic}, 171(9):102824, 31.
\newblock MR4100759.

\bibitem[Naturman and Rose, 1993]{nr93}
Naturman, C. and Rose, H. (1993).
\newblock Interior algebras: {S}ome universal algebraic aspects.
\newblock {\em J. Korean Math. Soc.}, 30:1 -- 24.
\newblock MR1217401.

\bibitem[N{\'e}meti, 1983]{nem83}
N{\'e}meti, I. (1983).
\newblock Surjectiveness of epimorphisms is equivalent to {B}eth definability
  property.
\newblock Manuscript, Alfréd Rényi Institute of Mathematics, Budapest.

\bibitem[Rasiowa and Sikorski, 1954]{rs54}
Rasiowa, H. and Sikorski, R. (1954).
\newblock On existential theorems in non-classical functional calculi.
\newblock {\em Fund. Math.}, 41:21--28.
\newblock MR0069111.

\bibitem[Rautenberg, 1980]{rau80}
Rautenberg, W. (1980).
\newblock Splitting lattices of logics.
\newblock {\em Arch. math. Logik}, 20(3-4):155--159.
\newblock MR0603335.
\end{thebibliography}
\end{document}